\documentclass[11pt]{amsart} 
\usepackage{amsmath,amssymb} 
\usepackage{graphicx}
\usepackage[utf8]{inputenc}
\usepackage[font=small,labelfont=bf]{caption}

\begin{document}

\newtheorem{thm}{Theorem}[section]
\newtheorem{lem}[thm]{Lemma}
\newtheorem{prop}[thm]{Proposition}
\newtheorem{coro}[thm]{Corollary}
\newtheorem{defn}[thm]{Definition}
\newtheorem*{remark}{Remark}

\numberwithin{equation}{section}

\newcommand{\Z}{{\mathbb Z}} 
\newcommand{\Q}{{\mathbb Q}}
\newcommand{\PP}{{\mathbb P}}
\newcommand{\R}{{\mathbb R}}
\newcommand{\C}{{\mathbb C}}
\newcommand{\N}{{\mathbb N}}
\newcommand{\FF}{{\mathbb F}}
\newcommand{\T}{{\mathbb T}}
\newcommand{\fq}{\mathbb{F}_q}

\newcommand{\fixmehidden}[1]{}

\def\scrA{{\mathcal A}}
\def\scrB{{\mathcal B}}
\def\cI{{\mathcal I}}
\def\scrD{{\mathcal D}}
\def\cL{{\mathcal L}}
\def\cM{{\mathcal M}}
\def\cN{{\mathcal N}}
\def\cQ{{\mathcal Q}}
\def\scrR{{\mathcal R}}
\def\scrS{{\mathcal S}}

\newcommand{\rmk}[1]{\footnote{{\bf Comment:} #1}}

\renewcommand{\mod}{\;\operatorname{mod}}
\newcommand{\ord}{\operatorname{ord}}
\newcommand{\TT}{\mathbb{T}}
\renewcommand{\i}{{\mathrm{i}}}
\renewcommand{\d}{{\mathrm{d}}}
\renewcommand{\^}{\widehat}
\newcommand{\HH}{\mathbb H}
\newcommand{\Vol}{\operatorname{vol}}
\newcommand{\area}{\operatorname{area}}
\newcommand{\tr}{\operatorname{tr}}
\newcommand{\norm}{\mathcal N} 
\newcommand{\intinf}{\int_{-\infty}^\infty}
\newcommand{\ave}[1]{\left\langle#1\right\rangle} 
\newcommand{\E}{\mathbb E}
\newcommand{\Var}{\operatorname{Var}}
\newcommand{\Cov}{\operatorname{Cov}}
\newcommand{\Prob}{\operatorname{Prob}}
\newcommand{\sym}{\operatorname{Sym}}
\newcommand{\disc}{\operatorname{disc}}
\newcommand{\CA}{{\mathcal C}_A}
\newcommand{\cond}{\operatorname{cond}} 
\newcommand{\lcm}{\operatorname{lcm}}
\newcommand{\Kl}{\operatorname{Kl}} 
\newcommand{\leg}[2]{\left( \frac{#1}{#2} \right)}  
\newcommand{\id}{\operatorname{id}}
\newcommand{\beq}{\begin{equation}}
\newcommand{\eeq}{\end{equation}}
\newcommand{\bsp}{\begin{split}}
\newcommand{\esp}{\end{split}}
\newcommand{\bra}{\left\langle}
\newcommand{\ket}{\right\rangle}
\newcommand{\diam}{\operatorname{diam}}
\newcommand{\supp}{\operatorname{supp}}
\newcommand{\sgn}{\operatorname{sgn}}
\newcommand{\re}{\operatorname{Re}}

\newcommand{\sumstar}{\sideset \and^{*} \to \sum}

\newcommand{\LL}{\mathcal L} 
\newcommand{\sumf}{\sum^\flat}
\newcommand{\Hgev}{\mathcal H_{2g+2,q}}
\newcommand{\USp}{\operatorname{USp}}
\newcommand{\conv}{*}
\newcommand{\dist} {\operatorname{dist}}
\newcommand{\CF}{c_0} 
\newcommand{\kerp}{\mathcal K}

\newcommand{\gp}{\operatorname{gp}}
\newcommand{\Area}{\operatorname{Area}}

\title[Quantitative equidistribution for Schr\"odinger operators]
{Quantitative equidistribution \\ of eigenfunctions for toral \\ Schr\"odinger operators}
 
\author{Henrik Uebersch\"ar}
\address{Sorbonne Universit\'e and Universit\'e de Paris, CNRS, IMJ-PRG, F-75006 Paris, France.}
\email{henrik.ueberschar@imj-prg.fr}
\date{\today}
\maketitle
 
\begin{abstract} 
We prove a quantum ergodicity theorem in position space for the eigenfunctions of a Schr\"odinger operator $-\Delta+V$ on a rectangular torus $\T^2$ for $V\in L^2(\T^2)$ with an algebraic rate of convergence in terms of the eigenvalue. 

A key application of our theorem is a quantitative equidistribution theorem for the eigenfunctions of a Schr\"odinger operator whose potential models disordered systems with $N$ obstacles. We prove the validity of this equidistribution theorem in the limit, as $N\to\infty$, under the assumption that a weak overlap hypothesis is satisfied by the potentials modeling the obstacles, and we note that, when rescaling to a large torus (such that the density remains finite, as $N\to\infty$) this corresponds to a size decaying regime, as the coupling parameter in front of the potential will decay, as $N\to\infty$.

We apply our result to Schr\"odinger operators modeling disordered systems on large tori $\T^2_L$ by scaling back to the fixed torus $\T^2$. In the case of random Schr\"odinger operators, such as random displacement models, we deduce an almost sure equidistribution theorem on certain length scales which depend on the coupling parameter, the density of the potentials and the eigenvalue. In particular, if these parameters converge to finite, nonzero values, we are able to determine at which length scale (as a function of these parameters) equidistribution breaks down. In this sense, we provide a lower bound for the Anderson localization length as a function of energy, coupling parameter and the density of scatterers.
\end{abstract} 

\section{Introduction}

In many problems in quantum physics one is led to study Schr\"odinger operators of the type
$$H_V=-\Delta+V,$$ where $V$ describes a complex system. For example, in the theory of disordered quantum systems it is common to apply a one body approximation to reduce the many body Schr\"odinger operator to a Schr\"odinger operator of the form
\beq\label{complex_Schr}
H_\Omega=-\Delta+\sum_{\omega\in\Omega}V(\cdot-\omega), 
\eeq
where $\Omega$ denotes the set of positions of the obstacles, and, to simplify the analysis, one usually assumes the potential $V$ to be continuous and compactly supported (although more general potentials may be considered).

In 1958, Anderson demonstrated \cite{A}, in the case of a lattice model, that in a regime where the disorder is sufficiently strong the quantum eigenfunctions become exponentially localized, thereby challenging the established picture at the time that the eigenfunctions should be delocalized. This phenomenon is today known as Anderson localization and has been proved rigorously in various models, discrete as well as continuous \cite{GMP,FS,GHK,BK}. An important question concerns the existence of a transition from a localized to a delocalized regime (Anderson transition) \cite{Mo}. The scaling theory of Anderson, Abrahams, Licciardello and Ramakrishnan \cite{AALR} predicts that no such transition occurs in dimension $d=2$, where the presence of disorder should suffice to localize all eigenfunctions. The case $d=2$ is critical in the sense that the localization length tends to be very large so that in practice (physical experiments and numerical simulations) a delocalization transition is often observed, because the localization length exceeds the physical dimensions of the system.

Whereas Anderson localization is by now well-understood at the mathematical level \cite{BK,GMP,FS,GHK}, a proof of Anderson delocalization remains elusive. Progress has, however, been made by Klein who gave the first rigorous result on spectral and dynamical delocalization on the Bethe lattice \cite{Kl}, by Aizenman-Warzel \cite{AW} for trees, and by Erd\"os-Salmhofer-Yau \cite{ESY} who proved quantum diffusion of the time-dependent random Schr\"odinger equation in a scaling limit. Moreover, Erd\"os-Yau and Erd\"os-Eng \cite{EY,EE} derived that the quantum dynamics is governed by a linear Boltzmann equation in the weak coupling and low density limits. Delocalization results for the eigenvectors were proved by Geisinger \cite{Ge} and Anantharaman-Sabri \cite{AnS1,AnS2}.
 
In the present article, we prove a quantitative equidistribution theorem for the eigenfunctions of a Schr\"odinger operator with an $L^2$ potential on a rectangular $2$-torus. Our result can, in some sense, be thought of as a position space analogue of the well-known quantum ergodicity theorem for the Laplace-Beltrami operator on Riemannian manifolds whose geodesic flow is ergodic \cite{Sn74,CdV85,Z87}. Rivi\`ere-Hezari \cite{RiHe17} proved a theorem which is a special case of our result for the Laplacian on flat tori, whereas Rudnick-Marklof \cite{RMa12} showed quantum ergodicity in position space (no rate) on rational polygons. For Schr\"odinger operators we note that quantum ergodicity theorems were proven for the special case of a delta potential in position space for rectangular 2-tori \cite{RU}, in phase space for square tori \cite{KU14}, while there is scarring in momentum space on diophantine tori \cite{KU17}. Yesha generalized these results to delta potentials on $3$-tori \cite{Y1,Y2}.

Although we state our main results for the case of real-valued potentials on rectangular tori, which is the most interesting case for our purposes, our proof should work for any rectangular domain with general boundary conditions. In particular, it may be applied to Schr\"odinger operators with quasi-periodic boundary conditions which arise via Bloch-Floquet theory and play an important role in the study of quantum transport in graphene, as discussed by Fefferman-Weinstein in \cite{FW}. We discuss such applications in section \ref{sec-BC}. 
 
We apply our main theorem to the case of Schr\"odinger operators with potentials describing a disordered system, such as \eqref{complex_Schr}. We study the limit, as the number of potentials $N$ on the fixed torus $\T^2$ tends to infinity, where the support of the potentials has diameter $O(N^{-1/2})$. Under a weak overlap hypothesis on the set of positions of the potentials, we derive an equidistribution theorem which survives in the limit, as $N\to\infty$. We note that when rescaling to a larger torus such that the density of the scatterers remains finite, then the coupling parameter will decay, as $N\to\infty$. So this corresponds to as {\em size decaying regime}. This application of our theorem is discussed in section \ref{sec-complex}. The proof of the main theorem is given in section \ref{sec-equi}. 

We derive a corollary which deals with the case of random potentials, where we prove that the eigenfunctions are uniformly distributed, which holds almost surely with respect to the joint distribution of the random positions and survives in the limit, as $N\to\infty$, where again we note that we are in a {\em size decaying regime}. Another corollary deals with the general (possibly strong disorder) regime, where we derive an equidistribution theorem on a large torus $\T^2_L$ which survives in the limit as the size of the torus tends to infinity, provided that energy, density of scatterers and the coupling parameter of the potential satisfy a certain constraint. In particular, equidistribution will break down at a certain length scale if all parameters remain bounded away from zero and finite. In this sense, we provide a lower bound for the localization length of a general random Schr\"odinger operator as a function of energy, coupling parameter and density of scatterers.

The methods used to obtain the results in the present article are limited to the case $d=2$. However, it is conceivable to apply lattice point methods in higher dimension to deal with Schr\"odinger operators on higher-dimensional tori. The distribution of lattice points in higher-dimensional thin annuli is much more regular than in dimension $2$, so that one might hope to improve our methods sufficiently to establish an equidistribution theorem which survives in the thermodynamic limit, where the coupling strength, density of scatterers and energy are nonzero and finite.

\subsection*{Acknowledgements}
The author's work was supported by the grant ANR-17-CE40-0011-01 (project SpInQS) of the French Nantional Research Agency (ANR) and a d\'el\'egation CNRS. The author is grateful to Jens Marklof and Gabriel Rivi\`ere for many helpful comments.

\subsection{Notation} 

Let $\T^2=\R^2/2\pi\cL_0$, where $\cL_0=\Z(a,0)\oplus\Z(0,1/a)$, $a>0$, is a rectangular unimodular lattice. We denote by $\Delta=\partial_x^2+\partial_y^2$ the Laplacian on $\T^2$. So the Laplace eigenvalues are values of the quadratic form $Q(\xi)=|\xi|^2=a^2x^2+a^{-2}y^2$, $(x,y)\in\Z^2$.

We denote the set of ordered Laplace eigenvalues, counted without multiplicity, as follows
\beq
\cQ=\{0=n_0<n_1<n_2<\cdots<\cdots<n_k<\cdots\to+\infty\}.
\eeq

If $V\in L^2(\T^2,\R)$, then the Schr\"odinger operator $-\Delta+V$ on the domain $H^2(\T^2)$ is essentially self-adjoint and has discrete spectrum (cf. section 2.2 in \cite{BBZ} for a detailed discussion). Throughout this paper we will assume that $V$ is real-valued which is the most interesting case for our purposes. However, our theorem can easily be adapted for complex-valued potentials and/or general boundary conditions. We briefly discuss this in subsection \ref{sec-BC}.

We will consider the eigenvalue problem 
\beq
(-\Delta+V)\psi_\lambda=\lambda\psi_\lambda, \quad \|\psi_\lambda\|_{L^2(\T^2)}=1.
\eeq

\subsection{Main theorem}
For a subset of the ordered set of Laplace eigenvalues, $\cQ'\subset \cQ$, we denote the following subset of the real line $$\Sigma:=\bigcup_{n_k\in \cQ'}(n_k,n_{k+1}).$$
We prove the following quantitative equidistribution theorem for the eigenfunctions of the Schr\"odinger operator $H_V=-\Delta+V$ on $\T^2$, where we recall $V\in L^2(\T^2)$.
\begin{thm}\label{equi_thm}
There exists a subset of full density $\cQ'\subset \cQ$ such that for any orthonormal basis of eigenfunctions of $H_V$, denoted by $\{\psi_k\}_{k\geq0}$, $H_V\psi_k=\lambda_k\psi_k$, and for any test function $a\in C^\infty(\T^2)$, we have for any $k$ such that $\lambda_k\in\Sigma$ the bound
\footnote{For two functions $f,g\geq 0$ we denote by $f\lesssim g$ that there exists a positive constant $C$ s. t. $f\leq Cg$.}
\beq\label{equi_bound}
\Big|\int_{\T^2} a |\psi_k|^2 d\mu-\frac{1}{4\pi^2}\int_{\T^2}a d\mu\Big| \lesssim_{a,\epsilon}
 \|V\|_{L^2(\T^2)}\lambda_k^{-(1-3\theta)/2+\epsilon}
\eeq
where the implied constant depends on a certain number of derivatives of $a$, $d\mu$ denotes the Lebesgue measure on $\T^2$ and $\theta$ denotes the exponent in the error term of the circle law. \footnote{The Gauss circle law asserts $\#\{\xi\in\cL \mid |\xi|^2\leq T\} = \pi T + O_\epsilon(T^{\theta+\epsilon})$.} 
\end{thm}
\begin{remark}
Under the assumption of the Gauss circle problem, which asserts that $\theta=1/4+\epsilon$, the best rate we can get is $O_\epsilon(\lambda_k^{-1/8+\epsilon})$. The best known exponent to date is due to Bourgain-Watt \cite{BW}, $\theta=517/1648+\epsilon$, so our theorem gives a rate of $O_\epsilon(\lambda_k^{-97/3296+\epsilon})$ at the time of writing.
\end{remark} 

\subsection{General boundary conditions, non-self-adjoint operators}\label{sec-BC}
We state theorem \ref{equi_thm} in the most natural case, for our purposes, of a self-adjoint operator on the torus $\T^2$. However, our proof may easily be adapted to deal with any rectangular domain and boundary condition as well as complex-valued potentials. In this case the same equidistribution theorem holds for eigenfunctions which are solutions to the eigenvalue problem
\beq
(-\Delta+V)\psi_\lambda=\lambda\psi_\lambda
\eeq
on a rectangular domain $\scrR=[0,2\pi a]\times[0,2\pi/a]$, provided the eigenvalue satisfies the condition $\re\lambda\in\Sigma$. The boundary conditions are in fact of no importance in the proof.

In order to adapt the proof one simply replaces the bound \eqref{Fourier_bound} with
$$|\hat\psi_\lambda(\xi)|^2\leq \|V\|_{L^2(\scrR)}^2(|\xi|^2-\re\lambda)^{-2}$$
and uses the assumption $\re\lambda\in\Sigma$ throughout the proof.

We thus have for any test function $a\in C^\infty(\T^2)$ and for any $k$ such that $\re\lambda_k\in\Sigma$ the bound
\beq
\Big|\int_{\T^2} a |\psi_k|^2 d\mu-\frac{1}{4\pi^2}\int_{\T^2}a d\mu\Big| \lesssim_{a,\epsilon} \|V\|_{L^2(\T^2)}(\re\lambda_k)^{-(1-3\theta)/2+\epsilon}.
\eeq

\subsection{Eigenfunctions in disordered systems}
Let us now consider the case of a potential which models a disordered system of $N$ scatterers of diameter $O(N^{-1/2})$ on $\T^2$
\beq\label{thermo_pot}
V_N(x)=\sum_{\omega_j\in\Omega_N} V(N^{1/2}(x-\omega_j))
\eeq
where $V\in L^2(\T^2)$, $\supp V\subset B(0,r)\subset\T^2$ and
$$\Omega_N:=\{\omega_1,\cdots,\omega_N\}\subset\T^2.$$

For fixed $N$ one may apply Theorem \ref{equi_thm} to the Schr\"odinger operator $H^N=-\Delta+V_N$ and one will obtain an error term which depends on $N$ (cf. section \ref{sec-complex}). So the equidistribution theorem is in general not valid in the limit, as $N\to+\infty$. 
 
The key question is whether, under a {\em weak overlap hypothesis} on the configuration of the set $\Omega_N$, one may strengthen the theorem in such a way that the error term obtained is {\em uniform with respect to the parameter $N$}. This means that, if the disorder in the system is sufficiently weak and the energy is above a certain threshold, then {\em delocalized eigenfunctions do exist} and their delocalization persists in the limit as $N\to\infty$, which means the delocalization is scale-invariant (we can think of the limit $N\to\infty$ as ``zooming'' out to larger and larger scales).

This question is answered in the affirmative by the following equidistribution theorem for the eigenfunctions of the Schr\"odinger operator $H^N=-\Delta+V_N$ which is valid in the limit, as $N\to+\infty$, if a certain weak overlap hypothesis is satisfied. The proof is given in section 
\ref{subsec-thermo}.
\begin{thm}\label{thermo}
Let us suppose that the set $\Omega_N$ satisfies 
\beq\label{rigid}
\#\Omega_N\cap B(x_0,\frac{R}{N^{1/2}}) \lesssim \max(1,R^2).
\eeq
for any $x_0\in\T^2$ and $0<R\lesssim N^{1/2}$. 
Let $H^N=-\Delta+V_N$, where $V_N$ is given by \eqref{thermo_pot}. Then for any orthonormal basis of eigenfunctions of $H_V$ denoted by $\{\psi_{N,k}\}_{k\geq0}$, $H^N\psi_k=\lambda_{N,k}\psi_{N,k}$, and for any test function $a\in C^\infty(\T^2)$, we have for any $k$ such that $\lambda_{N,k}\in\Sigma$ the bound
\beq\label{equi_bound_thermo}
\Big|\int_{\T^2} a |\psi_{N,k}|^2 d\mu-\frac{1}{4\pi^2}\int_{\T^2}a d\mu\Big| \lesssim_{a,\epsilon} r^2\|V\|_{L^2(\T^2)}(\lambda_{N,k})^{-(1-3\theta)/2+\epsilon}
\eeq
where we recall that $\supp V\subset B(0,r)$.
\end{thm}

\begin{remark}
Note that condition \eqref{rigid} is satisfied by distorted lattices. If we consider, for simplicity, $\Z^2\subset \R^2$ and we introduce finite displacements at each lattice point, we obtain the distorted lattice
$$\Omega=\{\xi+\omega_\xi \mid \xi\in\Z^2, \; |\omega_\xi|\leq r_0\}.$$
So the set $\Omega_N:=N^{-1/2}\Omega$ satisfies our condition for any ball $B(x_0,R/N^{1/2})$, because
$$|N^{-1/2}(\xi+\omega_\xi)-x_0|\leq R/N^{1/2}$$ implies
$$|\xi-N^{1/2}x_0|\leq |(\xi+\omega_\xi)-N^{1/2}x_0|+|\omega_\xi|\leq R+r_0$$
and the number of lattice points in this ball is $\sim\pi (R+r_0)^2$.
\end{remark}


\subsection{Random potentials}
We may apply Theorem \ref{thermo} to a popular model of disordered quantum systems, so-called random displacement models (RDM). 
It is important to note that the RDM for $N$ obstacles on $\T^2$ which we study here, when rescaled to a torus of size $N^{1/2}$, corresponds to a weak coupling strength $1/N$.
We will show that delocalization occurs for suffiently high energy in our RDM, in the sense that we prove an equidistribution theorem for the eigenfunctions, as $N\to+\infty$. 

Random Schr\"odinger operators such as those arising in RDM will satisfy condition \eqref{rigid} almost surely if the random variables satisfy a weak overlap assumption. 
Let $\PP$ be a compactly supported probability density on $\R^2$ where $\supp\PP\subset B(0,r_1)$ for some $r_1>0$. We consider a sequence $\{\Xi_j\}_{j=1}^{+\infty}$ of i.i.d. random variables with probability densities $\PP_j=\PP$. So the sequence $\{\Xi_j\}_{j=1}^{+\infty}$ has a joint probability density $\PP_\infty=\bigotimes_{j=1}^\infty\PP_j$. The proof of the following corollary of Theorem \ref{thermo} will be given in section \ref{subsec-random}.
\begin{coro}\label{weak}
Let $\bar{\Omega}_N=\{\bar\omega_1,\cdots,\bar\omega_N\}\subset\T^2$ be a fixed point set which satisfies \eqref{rigid}. $\Omega_N=\{\omega_1,\cdots,\omega_N\}$ be a random set, where $\omega_j=\bar\omega_j+\Xi_j$. Suppose $V\in L^2(\T^2)$ with $\supp V\subset B(0,r_2)$. 

We consider the random Schr\"odinger operator 
\beq
H^N=-\Delta+\sum_{\omega_j\in\Omega_N}V(N^{1/2}(x-\omega_j))
\eeq
Under the assumptions above we have $\PP_\infty$-a.s. that, for any orthornormal basis of eigenfunctions of the random Schr\"odinger operator $H^N$ and for any test function $a\in C^\infty(\T^2)$, for all eigenvalues $\lambda_k^N$ of $H^N$ such that $\lambda_k^N\in\Sigma$, the associated eigenfunctions satisfy
\beq\label{equi_bound}
\Big|\int_{\T^2} a |\psi_k^N|^2 d\mu-\frac{1}{4\pi^2}\int_{\T^2}a d\mu\Big| \lesssim_{a,\epsilon} (r_1+r_2)^2\|V\|_{L^2(\T^2)}(\lambda_k^N)^{-(1-3\theta)/2+\epsilon}.
\eeq
\end{coro}

Let us now consider the case of a Schr\"odinger operator on $\T^2_L=\R^2/(2\pi L\Z)^2$, where $L>0$ is a large parameter,
\beq
H_L=-\Delta+\alpha\sum_{j=1}^N V(x-x_j), \quad H_L\Psi=E_k\Psi_k, 
\eeq
$E_k$ denotes the energy, $\alpha=\alpha(L)$ a coupling parameter, and $\Omega_N=\{x_1,\cdots,x_N\}\subset\T^2_L$.
Let $\Omega\subset\R^2$ be a point set which is of positive density in the sense that
$$L^{-2}\#\Omega\cap[-L,L]^2 \asymp \rho$$ for some $\rho>0$. If we now apply periodic boundary conditions and take $\Omega_N=\Omega\cap 2\pi[-L,L]^2$, then this gives rise to a Schr\"odinger operator $H_L$ on $\T^2_L$. We then have the following corollary of Theorem \ref{equi_thm}, which we prove in section \ref{sec-strong}. In particular, me may take $\Omega$ to be a random point set which is generated by a suitable stochastic process (e. g. a Poisson process or a RDM).
\begin{coro}\label{strong}
For any orthornormal basis of eigenfunctions of the Schr\"odinger operator $H_L$: $H_L\Psi_k=E_k\Psi_k$, and for any test function $a\in C^\infty(\T^2)$, for all eigenvalues such that $E_k L^2\in\Sigma $, the associated eigenfunctions satisfy
\beq\label{equi_bound}
\Big|\int_{\T^2_L} a(\cdot/L) |\Psi_k|^2 d\mu-\frac{1}{(2\pi L)^2}\int_{\T^2_L}a(\cdot/L) d\mu\Big| 
\lesssim_{a,\epsilon} \rho|\alpha|\|V\|_{L^2(\T^2_L)}\frac{L^{2+3\theta+2\epsilon}}{E_k^{(1-3\theta)/2-\epsilon}}.
\eeq
\end{coro}
\begin{remark}
We note that in a regime, where\footnote{
We denote by $\ll$ ``much smaller than'', and by $\gg$ ``much larger than''.
}
\beq\label{equi_cond}
\rho|\alpha|\|V\|_{L^2(\T^2_L)}L^{2+3\theta+2\epsilon}\ll E_k^{(1-3\theta)/2-\epsilon},
\eeq
the equidistribution survives in the limit, as $L\to+\infty$. 

We note that in the thermodynamic limit, where $\alpha$, $\rho$ converge to fixed values and $E_k\in[E,2E]$, for some $E>0$, as $L\to+\infty$, the error term blows up due to the term $L^{2+3\theta+2\epsilon}$. This could be due to the torus reaching the scale of the localization length beyond which no equidistribution can hold.

In particular, this gives information on how the length scale of Anderson localization depends on the energy $E$, the density $\rho$ and the coupling parameter $\alpha$. Namely, the localization length grows with the energy and the inverse of the density and coupling parameter in the following way:
\beq
L_{loc}(\alpha,E,\rho,V)\gg \left(\frac{E^{(1-3\theta)/2-\epsilon}}{\rho|\alpha|\|V\|_{L^2(\T^2_L)}}\right)^{\frac{1}{2+3\theta+2\epsilon}}
\eeq 

On the other hand, if one allows $E=E(L)\to+\infty$ sufficiently fast, as $L\to+\infty$, then the localization length could grow faster than the size of the torus, which would be one interpretation of the regime corresponding to condition \eqref{equi_cond}. Another way this could be satisfied is if the $\rho=\rho(L)\to 0$ and/or $\alpha=\alpha(L)\to 0$ sufficiently fast, as $L\to+\infty$.

We stress that one may, of course, improve this lower bound significantly by exploiting the particular properties of the stochastic process which generates the point set $\Omega$. For example, in the case of RDM, we used assumption \eqref{rigid} to improve the term $L^{2+3\theta+\epsilon}$ to $L^{1+3\theta+\epsilon}$. Moreover, we were considering a weak coupling/low energy regime, where $\alpha=L^{-2}$, $E=L^{-2}E_1$, where $E_1$ is the fixed energy on $\T^2$. So, condition \eqref{equi_cond} becomes $\rho\|V\|_{L^2(\T^2)}\ll E_1^{(1-3\theta)/2-\epsilon}$ which shows that scale-invariant equidistribution holds, when $E_1$ is large enough with respect to the density of the scatterers and the $L^2$ norm of the potential.

We also point out that Klopp proved the existence of a localized regime for semi-classical random Schr\"odinger operators modeling random displacement models \cite{K}. However, there is no conflict with our result, as the regime considered by Klopp violates our condition \eqref{equi_cond} so that we cannot pass to the limit $L\to+\infty$, as our error term will blow up, which is consistent with Klopp's findings. To see this, note that Klopp considers semi-classical Schr\"odinger operators of the form $-\hbar\Delta+V$. So the regime of finite energy corresponds in our notation to a regime where $|\alpha|\asymp\hbar^{-2}\asymp E$. Substitution in \eqref{equi_cond} shows that the error term will blow up if we take $E\to+\infty$, as $L\to+\infty$.
\end{remark}


\section{Application to disordered systems}\label{sec-complex}

\subsection{Proof of Corollary \ref{strong}}\label{sec-strong} 
Let us consider Schr\"odinger operators of the type
\beq
H^N=-\Delta+V_N
\eeq
where $V\in L^2(\T^2)$ and
\beq
V_N(x)=\sum_{\omega_j\in\Omega_N} V(x-\omega_j).
\eeq
where $$\Omega_N:=\{\omega_1,\cdots,\omega_N\}\subset\T^2.$$

Given $\Omega_N$ let $\{\psi_j^N\}_{j\geq0}$, $H^N\psi_j^N=\lambda_j\psi_j^N$, be an orthonormal basis of eigenfunctions of the operator $H^N$. Theorem \ref{equi_thm} yields, for any $a\in C^\infty(\T^2)$ with $\int_{\T^2} a d\mu=0$, and for any $j$ such that $\lambda_j^N\in \Sigma$,
\beq  
|\bra a\psi_j^\omega,\psi_j^\omega\ket| \lesssim_a \|V_N\|_{L^2(\T^2)}(\lambda_j^N)^{-\beta}, \;\text{where} \;\beta=(1-3\theta)/2-\epsilon.
\eeq

Moreover, we have, for any $\Omega_N\subset\T^2$,
\beq\label{uni_bound}
\|V_N\|_{L^2(\T^2)}\leq N\|V\|_{L^2(\T^2)}.
\eeq 

Let us recall the setting of corollary \ref{strong}. 
If we scale back to the torus $\T^2=\R^2/(2\pi\Z)^2$, then we obtain the operator
\beq
h_L=-\Delta+\alpha L^2 \sum_{j=1}^N V(Lx-x_j), \quad h_L\psi=\lambda\psi, \quad \lambda=EL^2,
\eeq
where the orthonormal basis of eigenfunctions of $h_L$ on $L^2(\T^2)$ is obtained from the o.n.b. ${\Psi_k}$ on $L^2(\T^2_L)$ by scaling: $\psi_k(\cdot)=L\Psi_k(L\cdot)$ and $\lambda_k=E_k L^2$.

Hence, our potential $V_N$ is of the form (recall that $N\asymp \rho L^2$)
$$V_N=\alpha L^2\sum_{j=1}^N V(Lx-x_j)$$
so that we get 
\beq\label{pot_L2}
\begin{split}
\|V_N\|_{L^2(\T^2)}\leq N|\alpha|L^2\|V(L\cdot)\|_{L^2(\T^2)}
&=N|\alpha|L\|V\|_{L^2(\T^2_L)}\\
&\lesssim \rho|\alpha|L^3\|V\|_{L^2(\T^2_L)}
\end{split}
\eeq
which yields the error term on the r.h.s. in the corollary in view of $\lambda_k=E_kL^2$.

On the l.h.s. we have
$$|\bra a\psi_k,\psi_k\ket|=|L^2\int_{\T^2}a(x) |\Psi_k(Lx)|^2 dx|=|\int_{\T^2_L}a(y/L)|\Psi_k(y)|^2dy|.$$
We may now replace $a$ with $a-(4\pi^2)^{-1}\int_{\T^2} a d\mu$, so that, upon substituting and rewriting the constant term as $(4\pi^2L^2)^{-1}\int_{\T^2_L} a(\cdot/L) d\mu'$, we obtain the result.

\subsection{Proof of Theorem \ref{thermo}}\label{subsec-thermo}
An important question concerns the validity of our equidistribution theorem in the limit as $N\to+\infty$. Let us consider potentials whose support is of diameter $O(N^{-1/2})$. We assume $\supp V\subset B(0,r)$, for some $r>0$, and consider the potential
\beq
V_N(x)=\sum_{\omega_j\in\Omega_N}V(N^{1/2}(x-\omega_j))
\eeq

Let us assume that the set $\Omega_N$ satisfies the following {\bf weak overlap hypothesis}. We suppose that for any $x_0\in \T^2$ and $0<R\leq N^{1/2}$ we have
\beq
\Omega_N\cap B(x_0,\frac{R}{N^{1/2}}) \lesssim \max(1,R^2).
\eeq

Under this assumption, we have the following bound on the $L^2$ norm of $V_N$.
First of all, we may rewrite
$$
\int_{\T^2}|V_N|^2 d\mu =\sum_{\omega_j\in\Omega_N} \sum_{|\omega_k-\omega_j|\leq 2 r/N^{1/2}} \int_{\T^2} V(N^{1/2}x)\overline{V(N^{1/2}(x-\omega_j+\omega_k))}dx
$$
where we note that $|\omega_j-\omega_k|>2r/N^{1/2}$ implies 
$$|N^{1/2}(x-\omega_j+\omega_k)|\geq N^{1/2}|\omega_j-\omega_k|-N^{1/2}|x|>r$$
where we used $N^{1/2}x\in\supp V\subset B(0,r)$.

Now
\begin{align*}
&\Big|\int_{\T^2} V(N^{1/2} x)\overline{V(N^{1/2} (x-\omega_j+\omega_k))}dx\Big|\\
\leq& (\int_{\T^2} |V(N^{1/2} x)|^2dx)^{1/2}
(\int_{\T^2} |V(N^{1/2} (x-\omega_j+\omega_k))|^2dx)^{1/2}\\
=&\int_{\T^2} |V(N^{1/2} x)|^2dx=N^{-1}\|V\|_{L^2(\T^2)}^2
\end{align*}
because of translation invariance.

The weak overlap condition \eqref{rigid} now yields
\beq
\int_{\T^2}|V_N|^2 d\mu\lesssim \#\Omega_N N^{-1}\|V\|_{L^2(\T^2)}^2 r^2 = \|V\|_{L^2(\T^2)}^2 r^2
\eeq
so that the conditions of Theorem \ref{equi_thm} are satisfied, and, moreover, this yields an equidistribution bound \eqref{equi_bound_thermo} which is uniform with respect to $N$.

\subsection{Proof of Corollary \ref{weak}}\label{subsec-random}
If we take the positions $\omega_j$ to be independent random variables with probability densities whose compact supports are disjoint and of diameter $O(N^{-1/2})$, then we may ensure that the weak overlap hypothesis \eqref{rigid} is satisfied almost surely, and we may, thus, apply Theorem \ref{thermo} to obtain an equidistribution bound which is uniform with respect to $N$, and, therefore, valid in the thermodynamic limit, as $N\to+\infty$.

We first construct an infinite sequence of random variables on $\R^2$. Let $\PP$ be a compactly supported probability density on $\R^2$ where $\supp\PP\subset B(0,r_1)$ for some $r_1>0$. We consider a sequence $\{\Xi_j\}_{j=1}^{+\infty}$ of i.i.d. random variables with probability densities $\PP_j=\PP$. So the sequence $\{\Xi_j\}_{j=1}^{+\infty}$ has a joint probability density $\PP_\infty=\bigotimes_{j=1}^{+\infty}\PP_j$.


Let us suppose that $\bar\Omega_N\subset\T^2$ is a set of $N$ points on $\T^2$ which satisfy the weak overlap condition \eqref{rigid}. Let us introduce the random set $\Omega_N=\{\omega_j\}_{j=1}^N$, $\omega_j=\bar\omega_j+N^{-1/2}\Xi_j$.

It follows that $\Omega_N$ satisfies the following bound almost surely for any $x_0\in\T^2$, $R>0$:
$$\#\Omega_N\cap B(x_0,R/N^{1/2})\lesssim \max(1,(r_1+R)^2).$$
This follows from the observation that if $|\omega_j-x_0|\leq R/N^{1/2}$, then, almost surely,
$$|\bar\omega_j-x_0|\leq |\bar\omega_j-\omega_j|+|\omega_j-x_0|\leq (r_1+R)N^{-1/2}$$
and, therefore, almost surely,
$$\#\{j\mid |\omega_j-x_0|\leq R/N^{1/2}\}\leq \#\{j\mid |\bar\omega_j-x_0|\leq r_1+R\}\lesssim \max(1,(r_1+R)^2).$$


So $\Omega_N$ satisfies our weak overlap hypothesis almost surely.
Then, we have for an orthonormal basis of eigenfunctions $\{\psi_j^\omega\}_{j\geq0}$, for any $a\in C^\infty(\T^2)$ with $\int_{\T^2} a d\mu=0$, and for any $j$ such that $\lambda_j\in \Sigma$,
\beq
|\bra a\psi_j^\omega,\psi_j^\omega\ket| \lesssim_{a,\epsilon} (r_1+R)^2\|V\|_{L^2(\T^2)}\lambda_j^{-(1-3\theta)/2+\epsilon}.
\eeq

\section{Proof of Theorem \ref{equi_thm}}\label{sec-equi}
Let $a\in C^\infty(\T^2)$ with $\int_{\T^2}a d\mu=0$ and $L:=n_k^\delta$.

We have the following $L^2$ representation for the eigenfunctions 
\beq
\psi_\lambda=\sum_{\xi\in\cL}\hat{\psi}(\xi)e_\xi, \quad e_\xi(x):=\frac{1}{2\pi}e^{i\left\langle\xi,x\right\rangle}
\eeq
where $\cL$ denotes the dual lattice of $\cL_0$.
 
Moreover, for $\lambda\in (n_k,n_{k+1})$ let us define the following functions which are truncated outside an annulus 
\beq
\psi_{\lambda,L}(x)=\sum_{\xi\in A(n_k,L)}\hat\psi_\lambda(\xi)e_\xi(x)  
\eeq
where $L>0$ and $$A(n_k,L):=\{\xi \in \cL \mid ||\xi|^2-n_k|\leq L\}.$$

We have
\beq
\bra a\psi_\lambda,\psi_\lambda\ket=\bra a\psi_{\lambda,L},\psi_{\lambda,L}\ket+R(\psi_\lambda,L)
\eeq
where the remainder is of the form
$$R(\psi_\lambda,L):=\bra a\psi_{\lambda,L},\psi_{\lambda,L}^R\ket+\bra a\psi_{\lambda,L}^R,\psi_{\lambda,L}\ket+\bra a\psi_{\lambda,L}^R,\psi_{\lambda,L}^R\ket$$ where we denote $\psi_{\lambda,L}^R:=\psi_\lambda-\psi_{\lambda,L}$.

\subsection{Construction of $\cQ'$}

We will construct a full density subset $\cQ_1\subset \cQ$ such that for certain small $\delta>0$ and any $n\in \cQ_1$ the thin annuli $A(n,n^\delta)$ have the property that their lattice points are well-spaced. More precisely, we have the following proposition which we will prove in section \ref{sec-good}.
\begin{prop}\label{good_ann}
Let $\delta\in(\theta/2,1/2-\theta)$ and $\epsilon<\tfrac{1}{2}-\theta-\delta$. There exists a full density subset $\cQ_1\subset \cQ$ such that for any $n\in \cQ_1$ we have
that $\xi\in A(n,n^\delta)$ implies that for any $\zeta\in\cL\setminus\{0\}$, $|\zeta|\leq n^\epsilon$, we have $||\xi+\zeta|^2-n|\gg n^\delta$.
\end{prop} 

We recall (cf. \cite{RU}, Lemma 2.1) that the subset
$$\cQ_2:=\{n_k \in \cQ \mid |n_{k+1}-n_k|\leq n_k^{o(1)}\}$$
is of full density in $\cQ$, where $n_k^{o(1)}\to+\infty$, as $n_k\to+\infty$.
 
We define $\cQ':=\cQ_1\cap \cQ_2$.  

\subsection{Bound on $R(\psi_\lambda,L)$}
By Cauchy-Schwarz,
$$R(\psi_\lambda,L)\leq \|a\|_\infty(2 \|\psi_{\lambda,L}^R\|_{L^2(\T^2)}+\|\psi_{\lambda,L}^R\|_{L^2(\T^2)}^2),$$
so we have to estimate $\|\psi_{\lambda,L}^R\|_{L^2(\T^2)}$.

We have
$$\|\psi_{\lambda,L}^R\|_{L^2(\T^2)}^2=\sum_{\xi\notin A(n_k,L)}|\hat\psi_\lambda(\xi)|^2.$$
Recall that $\lambda\in(n_k,n_{k+1})$. From the identity $\psi_\lambda=(\Delta+\lambda)^{-1}(V\psi_\lambda)$ we deduce
\beq
\hat\psi_\lambda(\xi)=(-|\xi|^2+\lambda)^{-1}\int_{\T^2}e_\xi(-y)V(y)\psi_\lambda(y)dy
\eeq
and by Cauchy-Schwarz 
\beq\label{Fourier_bound}
|\hat\psi_\lambda(\xi)|^2 \leq \|V\|_{L^2(\T^2)}^2(|\xi|^2-\lambda)^{-2}.
\eeq

In view of this estimate, we have
\beq
\|\psi_{\lambda,L}^R\|_{L^2(\T^2)}^2=\sum_{\xi\notin A(n_k,L)}|\hat\psi_\lambda(\xi)|^2
\leq \|V\|_{L^2(\T^2)}^2\sum_{\xi\notin A(n_k,L)}(|\xi|^2-\lambda)^{-2}.
\eeq
Because $\xi\notin A(n_k,L)$, we have, by definition, 
$||\xi|^2-\lambda|\geq ||\xi|^2-n_k|-|n_k-n_{k+1}|>L-O(n_k^{o(1)})\gg n_k^\delta$
and, thus,
\beq
\|\psi_{\lambda,L}^R\|_{L^2(\T^2)}^2\leq \|V\|_{L^2(\T^2)}^2\sum_{||\xi|^2-\lambda|\gg n_k^\delta}(|\xi|^2-\lambda)^{-2}.
\eeq
This lattice sum can be bounded (cf. \cite{RU}, p. 11., eq. (5.9)) by a summation by parts which yields
\beq
\sum_{||\xi|^2-\lambda|\gg n_k^\delta}(|\xi|^2-\lambda)^{-2}\lesssim \frac{1}{L}+\frac{n_k^\theta}{L^2}
\eeq
where we recall $L=n_k^\delta$, so that $\min(\delta,2\delta-\theta)=2\delta-\theta$, because $\delta<1/2-\theta<\theta$ (recall $\theta> 1/4$), which gives $$\|\psi_{\lambda,L}^R\|_{L^2(\T^2)}^2\lesssim n_k^{-2\delta+\theta}.$$
Finally, $$R(\psi_\lambda,L)\lesssim \|a\|_\infty n_k^{-\delta+\theta/2}.$$

\subsection{Bound on $\bra a\psi_{\lambda,L},\psi_{\lambda,L}\ket$} 
Let us, first of all, prove the result for a monomial $a=e_\zeta$, $\zeta\in\cL\setminus\{0\}$.

We have
\beq
\bra e_\zeta \psi_{\lambda,L},\psi_{\lambda,L}\ket
=\sum_{\xi\in A(n_k,L)}\hat\psi_\lambda(\xi)\overline{\hat\psi_\lambda(\xi-\zeta)}
\eeq
and, thus, we obtain the bound
\beq
|\bra e_\zeta \psi_{\lambda,L},\psi_{\lambda,L}\ket| 
\leq \|V\|_{L^2(\T^2)}\left(\sum_{\xi\in A(n_k,L)}(|\xi-\zeta|^2-\lambda)^{-2}\right)^{1/2},
\eeq
where we used 
$$\sum_{\xi\in A(n_k,L)}|\hat\psi_\lambda(\xi)|^2\leq \|\psi_\lambda\|_{L^2(\T^2)}^2=1$$
as well as the bound \eqref{Fourier_bound}.

If $n_k\in \cQ'$, meaning that $A(n_k,L)$ is a good annulus, then, by definition, $\xi-\zeta\notin A(n_k,L)$, which yields the lower bound
$$||\xi-\zeta|^2-\lambda|\geq ||\xi-\zeta|^2-n_k|-O(n_k^{o(1)})\gg n_k^\delta.$$

And, hence, we obtain the bound
\beq
\begin{split}
|\bra e_\zeta \psi_{\lambda,L},\psi_{\lambda,L}\ket| 
&\leq \|V\|_{L^2(\T^2)} n_k^{-\delta} \# A(n_k,L)^{1/2}\\
&\lesssim \|V\|_{L^2(\T^2)} n_k^{-\delta+\theta/2}.
\end{split}
\eeq 

One may now readily generalize this bound for trigonometric polynomials of the form
$$a_\epsilon=\sum_{|\zeta|\leq n_k^\epsilon}\hat a(\zeta) e_\zeta,$$
namely,
\beq
\begin{split}
|\bra a_\epsilon \psi_{\lambda,L},\psi_{\lambda,L} \ket| 
\leq& \sum_{|\zeta|\leq n_k^\epsilon}|\hat a(\zeta)||\bra e_\zeta \psi_{\lambda,L},\psi_{\lambda,L}\ket| \\
\lesssim& \|\hat{a}\|_{\ell^1} \|V\|_{L^2(\T^2)} n_k^{-\delta+\theta/2}.
\end{split}
\eeq

\subsection{General test functions}\label{sec-general}
It remains to generalize our argument for $C^\infty$ test functions.
We have 
\beq
|\bra a\psi_{\lambda,L},\psi_{\lambda,L}\ket|
\leq |\bra a_\epsilon\psi_{\lambda,L},\psi_{\lambda,L}\ket|+|\bra(a-a_\epsilon)\psi_{\lambda,L},\psi_{\lambda,L}\ket|
\eeq
The first term is of order $O(n_k^{-\delta+\theta/2})$, as we saw above. 

For the second term, we have
$$|\bra(a-a_\epsilon)\psi_{\lambda,L},\psi_{\lambda,L}\ket|
\leq \sum_{|\zeta|>n_k^\epsilon}|\hat a(\zeta)|
\lesssim_{a,K} \sum_{|\zeta|>n_k^\epsilon} |\zeta|^{-K}
\lesssim_K n_k^{\epsilon(2-K)}$$
for any $K>2$.

In summary, we have (choose $\delta=1/2-\theta-\epsilon'$, for some $\epsilon'>0$)
\beq
|\bra a\psi_\lambda,\psi_\lambda\ket|\lesssim_{a,\epsilon} n_k^{-\delta+\theta/2}=n_k^{(-1+3\theta)/2+\epsilon'}.
\eeq
where the implied constant depends on a certain number of derivatives of $a$.
 
\section{Proof of Proposition \ref{good_ann}}\label{sec-good} 
In this section we will prove the following proposition.
\begin{prop}\label{good}
Let $\delta\in(\theta/2,1/2-\theta)$ and $\epsilon<\tfrac{1}{2}-\theta-\delta$. There exists a full density subset $\cQ_1\subset \cQ$ such that for any $n\in \cQ_1$ we have
that $\xi\in A(n,n^\delta)$ implies that for any $\zeta\in\cL\setminus\{0\}$, $|\zeta|\leq n^\epsilon$, we have $||\xi+\zeta|^2-n|\gg n^\delta$.
\end{prop}
\begin{proof}
Let us introduce the set of ``bad'' lattice points
\beq
S_\zeta:=\{\xi\in\cL \mid |\bra\xi,\zeta\ket|\leq |\xi|^{2\delta}\}.
\eeq


We define 
\beq
B_\zeta:=\{n\in \cQ \mid  A(n,n^\delta)\cap S_\zeta\neq\emptyset\}
\eeq
and
\beq\label{S1}
\cQ_1:=\{n\in \cQ \mid \forall\zeta\in\cL\setminus\{0\}, |\zeta|\leq n^\epsilon: A(n,n^\delta)\cap S_\zeta=\emptyset\}.
\eeq

We will adapt an argument from \cite{RU} to show that $B_\zeta$ is of density zero in $\cQ$ in the sense that\footnote{We note that in \cite{RU} the bad set was defined as a subset of a given interlacing sequence of the Laplace eigenvalues. Here the bad set is a subset of the Laplace eigenvalues. But the proof works in an analogous fashion.}
\beq\label{bad_ev_bound}
\# \{n\in B_\zeta \mid n\leq X\} \lesssim \frac{X^{1/2+\theta+\delta}}{|\zeta|}
\eeq
where $\delta\in(\theta/2,1/2-\theta)$.

To see this, we first of all recall the following bound (cf. Lemma 6.2 in \cite{RU}) on the number of bad lattice vectors in a ball of radius $X^{1/2}$
\beq\label{badset_bound}
\#\{\eta\in S_\zeta \mid |\eta|^2\leq X\} \lesssim \frac{X^{1/2+\delta}}{|\zeta|}.
\eeq
We will now use this bound to obtain a bound on the number of elements $n\in B_\zeta$, $n\leq X$. 

Let us denote
\beq
\cQ_\zeta:=\{q\in\R_+ \mid \exists\eta\in S_\zeta: q=|\eta|^2\}.
\eeq

To this end, we define the projection map $\pi: B_\zeta\to \cQ_\zeta$ which assigns to $n\in B_\zeta$ the element of $\cQ_\zeta$ which is closest to $n$. If there are two such elements, then we define $\pi(n)$ as the larger of these two elements.

For a given $m\in \cQ_\zeta$ we have that the preimage satisfies
$$\pi^{-1}(m)\subset \{n\in B_\zeta \mid \exists\eta\in S_\zeta\cap A(n,n^\delta): |\eta|^2=m\}.$$
To see the inclusion, suppose $n\in\pi^{-1}(m)$. Recall that by definition of $\cQ_\zeta$ we have some $\eta_0\in S_\zeta$ s. t. $|\eta_0|^2=m$. Suppose for a contradiction that $\eta_0\notin S_\zeta\cap A(n,n^\delta)$, so $\eta_0\notin A(n,n^\delta)$, which means $||\eta_0|^2-n|>n^\delta$. On the other hand, $m=\pi(n)$ and so any other lattice vector $\eta\in S_\zeta$ must satisfy $||\eta|^2-n|\geq|m-n|>n^\delta$, by definition of the map $\pi$. This means $\eta\notin A(n,n^\delta)$, hence $S_\zeta\cap A(n,n^\delta)$ is empty, a contradiction to $n\in B_\zeta$. It follows that $\eta_0\in S_\zeta\cap A(n,n^\delta)$ which proves the inclusion.

Moreover, we have (for $m$ large enough) 
\beq\label{incl}
\{n\in B_\zeta \mid \exists\eta\in S_\zeta\cap A(n,n^\delta): |\eta|^2=m\}\subset \cQ\cap[m-2m^\delta,m+2m^\delta]
\eeq
which easily follows, because for any $n\in B_\zeta$ in the set above we have $\eta\in A(n,n^\delta)$ with $|\eta|^2=m$. This means $|n-m|=|n-|\eta|^2|\leq n^\delta$ which implies $|n-m|\leq 2m^\delta$ for $m$ large enough.

We have
$$\#\{n\in B_\zeta \mid n\leq X\}\leq\sum_{m\in \cQ_\zeta}\#\pi^{-1}(m)\cap[0,X].$$

Now if $n\in\pi^{-1}(m)\cap[0,X]$, then $m=\pi(n)$ and $|n-m|\leq n^\delta$, because $S_\zeta\cap A(n,n^\delta)$ is non-empty, because $n\in B_\zeta$, and by definition $\pi(n)$ is the element of $\cQ_\zeta$ closest to $n$. So $m\leq n+n^\delta\leq X+X^\delta$, because $n\leq X$.

Hence, 
$$\sum_{m\in\cQ_\zeta}\#\pi^{-1}(m)\cap[0,X] \leq \sum_{\substack{ m\in\cQ_\zeta \\ m\leq X+X^\delta}}\#\pi^{-1}(m)
\leq \sum_{\substack{ m\in\cQ_\zeta \\ m\leq X+X^\delta}}\#\cQ\cap [m-2m^\delta,m+2m^\delta]$$
where we used the inclusion \eqref{incl}.

Define the multiplicity of the Laplace eigenvalue $n\in\cQ$ as
$$r_\cL(n)=\#\{\xi\in\cL \mid n=|\xi|^2\}.$$
The asymptotic $$\sum_{n\leq X} r_\cL(n)=\pi X+O(X^\theta)$$ yields (recall $\delta<\theta$)
$$\#\cQ\cap [m-2m^\delta,m+2m^\delta]\leq \sum_{|n-m|\leq 2m^\delta}r_\cL(n)\lesssim m^\delta+m^\theta\lesssim m^\theta.$$

Therefore, 
$$\sum_{\substack{ m\in\cQ_\zeta \\ m\leq X+X^\delta}}\#\cQ\cap [m-2m^\delta,m+2m^\delta]
\lesssim X^\theta \sum_{\substack{ m\in\cQ_\zeta \\ m\leq X+X^\delta}}1\lesssim \frac{X^{1/2+\delta+\theta}}{|\zeta|}$$
where we used \eqref{badset_bound}.

This yields the bound \eqref{bad_ev_bound}.

Let us now use \eqref{bad_ev_bound} to prove that the complement of $\cQ_1$ is of density zero in $\cQ$:
\beq\label{compl-bound}
\begin{split}
&\#\{m\in \cQ_1^c \mid m\leq X\}\\
=&\#\{m\in \cQ, m\leq X
\mid \exists \zeta\in\cL\setminus\{0\}, |\zeta|\leq m^\epsilon: A(m,m^\delta)\cap S_\zeta\neq\emptyset\}\\
\leq&\sum_{\substack{\zeta\in\cL\setminus\{0\} \\ |\zeta|\leq X^\epsilon}}
\#\{m\in B_\zeta \mid |m|\leq X\}\\
\lesssim& X^{1/2+\delta+\theta} \sum_{\substack{\zeta\in\cL\setminus\{0\}\\|\zeta|\leq X^\epsilon}}\frac{1}{|\zeta|}
\lesssim X^{1/2+\delta+\theta+\epsilon}=o(X).
\end{split}
\eeq

Let us now assume $n\in \cQ_1$ and $\xi\in A(n,n^\delta)$. We have to show that for any $\zeta\in\cL\setminus\{0\}$, $|\zeta|\leq n^\epsilon$, we have $||\xi+\zeta|^2-n|\gtrsim n^\delta$.

Indeed, 
\beq
\begin{split}
||\xi+\zeta|^2-n|&\geq 2|\bra\xi,\zeta\ket|-||\xi|^2-n|-|\zeta|^2\\
&\geq 2|\xi|^{2\delta}-n^\delta-O(n^{2\epsilon})\\
&\gtrsim n^{\delta}
\end{split}
\eeq
because $\xi\in A(n,n^\delta)$.

\end{proof}

\end{document}